\documentclass[ba, preprint]{imsart}
\RequirePackage{graphicx} 
\usepackage{amssymb}
\usepackage{bm}
\usepackage{amsmath}
\usepackage{algorithm}
\usepackage[noend]{algpseudocode}
\usepackage{dsfont} 
\usepackage{amsthm}
\usepackage{amsfonts}
\usepackage{bbold}
\usepackage{mathrsfs} 
\usepackage{multirow}
\usepackage{array}
\usepackage{subcaption}
\usepackage{mathtools}
\usepackage[authoryear]{natbib}
\RequirePackage[colorlinks,citecolor=blue,urlcolor=blue,backref=page,backref=page]{hyperref}
\usepackage{booktabs}

\newcommand{\Lam}{\Lambda}
\newcommand{\logit}{\text{logit}}
\usepackage[dvipsnames]{xcolor}

\usepackage{amsthm}
\usepackage{amsmath}
\usepackage{amssymb,bbm}
\usepackage{caption}
\pubyear{2025}
\arxiv{}
\volume{TBA}
\issue{TBA}
\firstpage{1}
\lastpage{1}

\startlocaldefs
\theoremstyle{plain}

\newtheorem{theorem}{Theorem}[section]

\theoremstyle{definition}

\theoremstyle{remark}

\endlocaldefs

\begin{document}

\begin{frontmatter}
\title{A Dependent Feature Allocation Model Based on Random Fields}
\runtitle{A dependent feature allocation model}

\begin{aug}
\author[A]{\fnms{Bernardo}~\snm{Flores}\ead[label=e1]{bernardoflo@utexas.edu}\orcid{0009-0003-7588-0663}},
\author[A]{\fnms{Yang}~\snm{Ni}\ead[label=e2]{yang.ni@austin.utexas.edu}\orcid{0000-0003-0636-2363}},
\author[B]{\fnms{Yanxun}~\snm{Xu}\ead[label=e3]{yanxun.xu@jhu.edu}\orcid{0000-0001-5554-8637}}
\and
\author[A]{\fnms{Peter}~\snm{Müller}\ead[label=e4]{pmueller@math.utexas.edu}\orcid{0000-0002-2948-1229}}
\address[A]{Department of Statistics and Data Sciences,
University of Texas at Austin\printead[presep={,\ }]{e1,e2,e4}.}
\address[B]{Department of Applied Mathematics and Statistics,
John Hopkins University\printead[presep={,\ }]{e3}}
\runauthor{Flores, B. et al.}
\end{aug}

\begin{abstract}
We introduce a flexible framework for modeling dependent feature
allocations. Our approach addresses limitations in traditional
nonparametric methods by directly modeling the logit--probability
surface of the feature paintbox, enabling the explicit incorporation
of covariates and complex but tractable dependence structures.  

The core of our model is a  Gaussian Markov Random Field (GMRF),
which we use to robustly decompose the latent field, separating a 
structural component based on the baseline covariates from intrinsic,
unstructured heterogeneity. This structure is not a rigid grid but a
sparse k--nearest neighbors graph derived from the latent
geometry  in the data, 
ensuring high--dimensional tractability. We extend this
framework to a dynamic spatio--temporal process, allowing item effects
to evolve via an Ornstein--Uhlenbeck process. Feature correlations are
captured using a low-rank factorization of their joint prior. We
demonstrate our model's utility by applying it to a polypharmacy
dataset, successfully inferring latent health conditions from patient
drug profiles. 
\end{abstract}



\end{frontmatter}
\section{Introduction}

We develop a novel feature allocation (FA) model for dependent
data. Instead of relying on 
a latent random measure, we directly model the
probability surface of the feature paintbox. Our method embeds the
feature allocation onto the unit square, $[0,1]^2$, transforming the
discrete feature assignment problem into the estimation of a
continuous, spatially-varying random field $\Lambda$. We model 
 $\Lambda$ using a hierarchical Bayesian spatial model that flexibly
captures dependencies along multiple axes. 

The proposed approach 
is designed to address key challenges within the broader
context of exchangeable feature allocations.
FA models provide a powerful framework for discovering latent
subsets (or properties) in data, with applications ranging from topic modeling to bioinformatics. Standard methods, often rooted in non-parametric
Bayesian approaches, such as the Indian Buffet Process
\citep{thibaux_beta}, model feature assignments implicitly.
While theoretically elegant, they face two major limitations: computational scalability for large datasets and, crucially,
lack of flexibility in explicitly modeling dependencies based on covariates like time, spatial proximity, or any other characteristic of the observations. Incorporating such structures often requires
complex, bespoke model extensions that are computationally burdensome
and difficult to generalize.

A key innovation of our model is the decomposition of observation--level
effects into a structured component, governed by a
conditional autoregressive (CAR) model, and an unstructured component
that captures independent heterogeneity. This formulation, known as 
Besag, York, and Mollié (BYM) model \citep{besag1991bym}, provides a
robust and interpretable framework for separating structured correlation
from random noise. This modular structure makes it straightforward to
introduce additional factors, such as dependencies between features
(e.g., medical conditions in our application) or temporal dynamics.

We demonstrate the utility and scalability of this approach with an application to polypharmacy in clinical data. Polypharmacy, the simultaneous use of multiple medications, presents a significant statistical challenge: a patient's prescription list is often a noisy proxy for their true health status due to symptom masking, where a drug prescribed for one condition inadvertently treats the symptoms of a separate, undiagnosed comorbidity. This disconnect makes it difficult to observe the true feature allocation directly. Our model addresses this by inferring the latent comorbidity profile from observed drug prescriptions, utilizing the shared structure among similar patients to unmask hidden conditions.

\section{Feature Paintbox and Feature--Frequency Models}
FA models are probability models for random
subsets of a set of observations (items). Subsets are
interpreted as features (or properties) of the items. 
FA generalizes clustering,
which restricts the random subsets to a partition. 
In contrast, in an FA, an item is not limited to one subset but can possess any number of features, including, no feature at all. 
For instance, in topic modeling, a document (an observation) can be linked to several different topics (features);
in genomics, a patient (an observation) can have multiple
co--occurring genetic markers (features). The goal of these models is
to infer this latent feature structure from the observed data, often
when the total number of features is unknown. 

Generatively, feature allocation models describe a stochastic process
by which each item acquires its set of features. These models
are often formulated within a Bayesian non--parametric framework, such
as the Indian Buffet Process (IBP) \citep{thibaux_beta}.
It can be defined 
as a prior over binary matrices (which encode feature membership)
with rows corresponding to items and columns to features,
allowing for a potentially infinite number of columns.
This allows the number of latent features in the data to be
inferred rather than specified in advance. A key property of many such
models is exchangeability, meaning the probability of an allocation is
invariant to the ordering of the items or features. 
This property is guaranteed by an extension of de Finetti's theorem,
which implies the existence of an underlying measure that
governs the generative process.

Exchangeable feature allocations have a representation that
generalizes Kingman's paintbox \citep{kingman1978representation} for random partitions,
which associates the clusters with mutually non-overlapping intervals
that partition the unit interval $[0,1]$.  The length of each interval
represents the prevalence of the corresponding cluster. In contrast,
Kingman's paintbox model for an exchangeable FA associates each
latent feature $c$ with a measurable, possibly disconnected subset $C_c
\subseteq [0,1]$, without the restriction
to mutually exclusive subsets and without restricting the union to
cover the entire unit interval.  
One may then generate a realization of the FA
by first drawing, for each observation $i$ an independent
uniform random variable $v_i$ on $[0,1]$ and assigning
$i$ to feature $c$ if $v_i \in C_c$ 
(see Figure \ref{fig:paintbox}). In this way, the law of
the collection $\{C_c\}_{c\ge1}$, known as the feature paintbox,
serves as the de Finetti mixing measure for the feature allocation. An
allocation admits an \emph{exchangeable feature probability function}
(EFPF) if and only if its distribution can be expressed solely in
terms of the sizes of the feature subsets, mirroring the role of the
exchangeable partition probability function (EPPF) in clustering
\citep{broderick2013feature}.

\begin{figure}[!ht]
    \centering
    \includegraphics[width=0.8\linewidth]{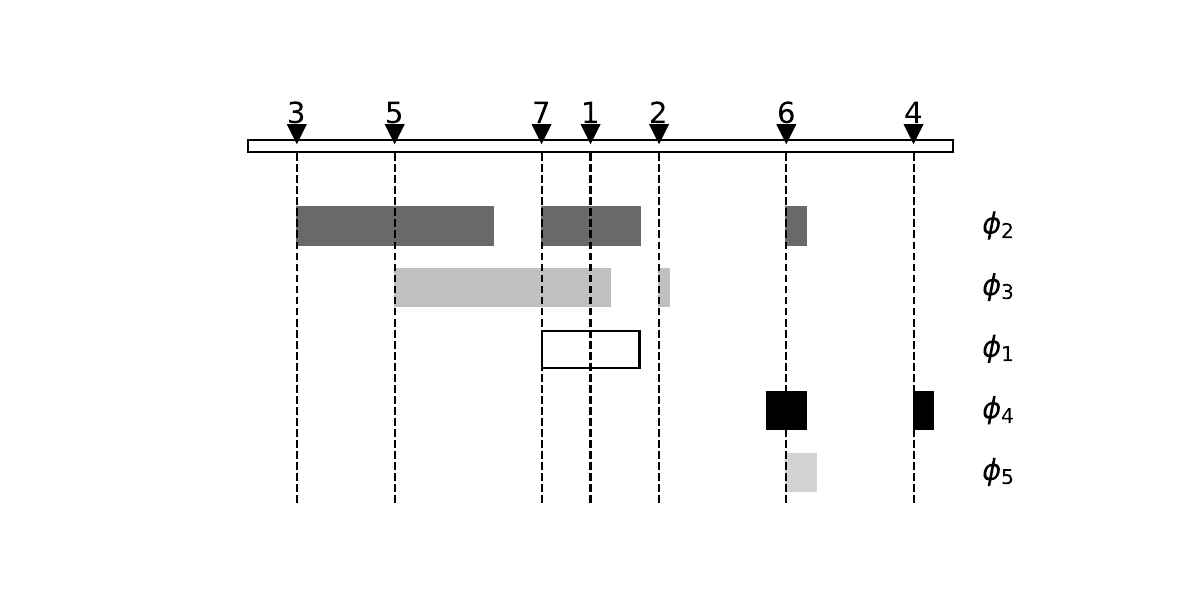}
    \caption{Visual representation of a feature paintbox. Image recreated from \cite{tamaratraits}.}
    \label{fig:paintbox}
\end{figure}

A particularly tractable subclass of feature allocation models is
provided by feature-frequency models. In these models, each
feature $c$ is assumed to have an associated random frequency $p_c$ and a label
$\phi_c$ drawn independently from a continuous distribution (typically
Uniform$(0,1)$). The generative process can be conceptualized in two
equivalent ways. For each item, feature $c$ can be seen as
included via an independent Bernoulli trial with probability
$p_c$.
The feature frequencies $p_c$ are encoded in a 
random measure
\[
B = \sum_{c=1}^\infty p_c\,\delta_{\phi_c},
\]
with the condition that $\sum_{c=1}^\infty p_c < \infty$ almost
surely, which ensures that each item exhibits only finitely
many features.

An alternative representation is provided by the feature paintbox, constructed via a recursive fractal slicing of the unit interval $[0,1]$. Let $\mathcal{A}_{c-1}$ denote the partition generated by the first $c-1$ features, where each atom $A \in \mathcal{A}_{c-1}$ represents a unique feature history defined by the intersection $A = \bigcap_{k=1}^{c-1} C_k^{\delta_k}$, with $C_k^1 = C_k$ and $C_k^0 = [0,1] \setminus C_k$ denoting the set and its complement respectively. To ensure independence with marginal probabilities $p_c$, the set $C_c$ is defined as the union of the leading $p_c$-proportion of each atom:
\[
C_c = \bigcup_{A \in \mathcal{A}_{c-1}} \operatorname{interval}(A, p_c).
\]
For each item $i$, a uniform variable $U_i \sim \operatorname{Uniform}[0,1]$ determines feature membership via $U_i \in C_c$. This construction explicitly realizes the independent Bernoulli trials, guaranteeing that the allocation distribution depends solely on $\{p_c\}$ while remaining exchangeable. Moreover, this formulation facilitates tractable inference, as integrating over the random measure $B$ yields closed-form predictive distributions and elucidates the correlation structure among latent features \citep{broderick2013feature,teh2009indian}.

\section{A dependent feature allocation model}
\subsection{Dependent features}
The feature frequency models described in Section 2 provide a
generative framework for exchangeable feature allocations. In a finite
setting with $M$ features, the model is fully specified by a vector of
feature probabilities, $p = (p_1, \dots, p_M)$, where $p_c$ is the
frequency of feature $c$. For any observation or item $i$, the feature
indicators $A_{ic}$ are drawn independently from a Bernoulli
distribution:
$$
A_{ic} \mid p_c \overset{\text{ind}}{\sim} \text{Bernoulli}(p_c),
\quad \text{for } c=1, \dots, M.
$$
In commonly used models, such as the IBP, one assumes that
the feature probabilities $p_c$
are drawn independently from a prior distribution, such as a Beta
distribution. This assumption of a priori independent
 feature frequencies $p_c$  
prevents the model from capturing known or suspected correlations
between features. For instance, in a medical context, the latent
probability of having diabetes should be correlated with the
probability of having hypertension.

To overcome this limitation, we define a parametric, dependent feature
frequency model. The core idea is to place a joint prior distribution
on the entire vector of feature probabilities
that explicitly encodes a dependency structure
(discussed in the next section). 
As probabilities are constrained to the
$[0,1]$ interval, we work on the logit scale. Let $\Lambda_c =
\text{logit}(p_c)$ be the latent logit-probability for feature $c$. We
then place a prior on the vector $\Lambda = (\Lambda_1, \dots,
\Lambda_M)$ that encourages correlation among related features. 

For example, one could 
impose a Conditional Autoregressive (CAR) model as a
prior on the feature paintbox, as represented by $\Lambda$. A CAR
model is a type of Gaussian Markov Random Field (GMRF) defined by a
neighborhood graph, $\mathcal{G}_F = (\mathcal{V}_F, \mathcal{E}_F)$,
where the vertices $\mathcal{V}_F = \{1, \dots, M\}$ represent the
features. An edge between features $c$ and $d$ indicates an a priori
belief that their latent probabilities are related. The CAR prior
specifies that the conditional distribution of any $\Lambda_c$, given
all other logit-probabilities, depends only on its neighbors
$\Lambda_{\partial_c}$: 
\begin{equation}\label{CAR}
p(\Lambda_c \mid \Lambda_{-c}) = p(\Lambda_c \mid
\Lambda_{\partial_c}) = N\left( \sum_{d \in \partial_c} w_{cd}
\Lambda_d, \tau_c^2 \right),
\end{equation}
where $w_{cd}$ are weights defining the strength of interaction and
$\tau_c^2$ is a conditional variance. The joint distribution for this
prior is a multivariate normal with a sparse precision matrix
determined by the graph $\mathcal{G}_F$. Conditions on the matrix of weights $w_{cd}$ forcing symmetry of the covariance matrix of the random surface ensure that
\eqref{CAR} defines a valid joint probability model on $\Lambda$
\citep{ver2018relationship}.
This approach allows for a construction of dependent features, 
which is essential when, for example, features define related medical conditions. In the upcoming discussion, we will add one more level of
generalization to also allow for the representation of dependence
across observations. 

The approach stands in contrast to conventional finite feature
models, such as the finite Indian Buffet Process. In the IBP
framework, inter-feature dependencies are induced implicitly through a
"rich-get-richer" sequential generative process, where the probability
of a new item acquiring a feature depends on its existing
popularity. While this induces a dependency, the structure is a
byproduct of the process and offers little direct control. Our
parametric formulation, by contrast, provides explicit control over
inter-feature dependencies. By placing a GMRF prior on the
logit-probabilities, the correlation structure is encoded directly
into the precision matrix via the neighborhood graph
$\mathcal{G}_F$. This allows a researcher to inject domain
knowledge—for instance, specifying that two medical conditions are
likely to be comorbid by creating an edge between them in the graph—in
a manner that is not straightforward in the IBP framework. This
explicit parameterization provides a more interpretable and flexible
tool for modeling complex systems where prior structural information
is available. 

\subsection{
  Dependent items -- 
  a standarized representation of the surface \texorpdfstring{$\Lambda$}{Lambda}}  

The parametric model introduced above can capture
dependencies between features, but still leaves the feature allocation of items to features
unspecified. 
 Independent sampling with the feature probabilities $p_c$ would
limit the ability to model scenarios where 
items themselves have a dependence structure. For example, 
patients could be related by demographic similarity, or items could be ordered in time. 

To incorporate both feature-level and item-level dependencies
simultaneously, we extend the parametric approach
 by augmenting $\Lam_c$ to $\Lam_{c,i} = \logit p_{ci}$ for the
probability of including item $i$ in feature $c$.
That is, we are 
embedding the entire feature allocation into the unit
square, $[0,1]^2$.
This transformation reframes the problem from modeling a single vector
of feature parameters to modeling a probability surface or random
field, making it directly amenable to spatial statistical models that
can capture dependencies along multiple axes. 

Inspired by spatial embedding methods for graph representations
\citep{borgsgraphons}, we propose a technique for embedding feature
data into the unit square. Consider a data set comprising $N$
items of $M$ features, $(A_{i1}, \dots, A_{iM})_{i=1}^N$.
\ In the construction of the desired embedding, we use the
vertical y-axis of the unit square to correspond to the features, and
the horizontal x-axis to represent the observations. 
We first partition the y-axis (representing
features) into $M$ disjoint intervals, defining $M$ horizontal strips
$S_1, \dots, S_M$, where $S_c$ uniquely corresponds to
feature $c$. Next, we partition the x-axis (representing the
item index
) into $N$ intervals,
defining $N$ vertical strips $S'_1, \dots, S'_N$. See Figure~\ref{fig:paintemb}. The ordering of
these strips can be informed by prior knowledge, such as patient
covariates or timestamps. 
We will argue that the use of equally spaced grids on vertical and
horizontal axis can be done without loss of generality.

The intersection of these strips defines a grid of $M \times N$
rectangular regions (pixels) within $[0,1]^2$. We then define a
function $\Lambda: [0,1]^2 \to \mathbb{R}$, which represents the
logit-probability surface of the feature allocation. For any point
$(x, y) \in [0,1]^2\in S_c \cap S'_i$, we
model $\Lambda(x, y) = \Lambda_{ic}$, the logit-probability that
feature $c$ is present in item $i$. This yields a piecewise
constant surface that represents the complete feature allocation
probabilities. Conditioned on this surface $\Lambda$, the feature
assignments $A_{ic}$ are independent Bernoulli trials. This embedding
is also dynamic; as new items become available, corresponding
vertical strips can be added to extend the representation. An example
of such an embedded feature paintbox is shown in
Figure~\ref{fig:paintemb}. 

\begin{figure}[h]
    \centering
    \includegraphics[width=0.7\linewidth]{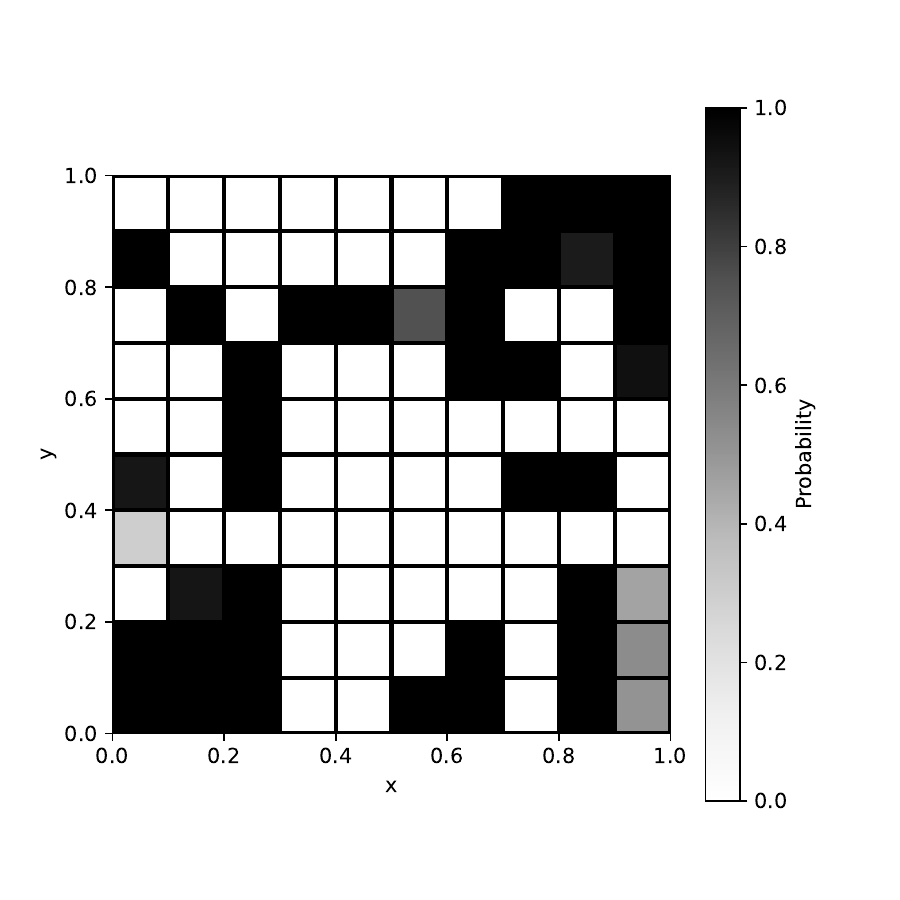}
    \caption{Embedding of a feature paintbox onto a probability surface $\Lambda$ in $[0,1]^2$.}
    \label{fig:paintemb}
\end{figure}

This representation of the feature allocation can be viewed as a
random field on a discrete grid. We define an undirected graph
$\mathbf{G}=\{\mathbf{V,\mathbf{E}}\}$ to represent the assumed
dependencies between the pixels in this spatial embedding. This
graphical representation will be used to define a spatial model for
the feature allocations. 

\section{Parametrization of the random surface}
\subsection{An interpretable hierarchical prior}
We model the latent logit--probability surface, $\Lambda$, using a hierarchical Bayesian framework that decomposes the field into item-- and feature--specific components. The logit--probability for item $i$ having feature $c$ is modeled as:
\[
\text{logit}(p_{ic}) = \Lambda_{ic} = \phi_i + \delta_c
\]
Here, $\phi_i$ is a latent effect for each item that captures
their baseline propensity and similarity to other items (e.g., patients), and
$\delta_c$ is an individual effect for each feature representing its
overall prevalence. This separates the column from the row
effects, and is a natural generalization of the Beta--Bernoulli
process, which assumes only feature--specific probabilities and iid
sampling for each observation; instead, we allow for more structure by also
specifying observation--specific probabilities. To model these observation--level
dependencies (rows), we use a Besag, York, and Mollié (BYM) model for
the patient effects $\phi$ \citep{besag1991bym}. This approach
robustly separates spatially structured correlation from unstructured
heterogeneity by defining $\phi_i$ as the sum of a structured spatial
component, $u_i$, and an unstructured random effect, $v_i$. 
\[
\phi_i = u_i + v_i
\]
The structured component, $u$, is given an Intrinsic Conditional
Autoregressive (ICAR) prior \citep{besagpseudo}. The ICAR prior is a
type of GMRF defined on an observation neighborhood graph
$\mathcal{N}_p$ (e.g., from k--NN), and it favors
neighboring observations to be similar according to their baseline
covariates. Its precision matrix is given by the graph Laplacian, $Q_u
= \tau_s (D_p - W_p)$, where $W_p$ is the adjacency matrix of the
observation graph and $D_p$ is the diagonal degree matrix. This prior
is improper, as it only defines relative differences between
observations. 

The unstructured component, $v$, consists of independent and identically distributed Gaussian random effects, $v_i \sim \mathcal{N}(0, 1/\tau_u)$, which capture observation--specific variation not explained by the spatial structure.

The combination of these two components results in a proper prior for the overall observation effect $\phi$, with a full-rank precision matrix $Q_\phi = \tau_s(D_p - W_p) + \tau_u I$, making it 
\begin{equation}
    \phi \sim N\left\{0,\, [\tau_s(D_p - W_p) + \tau_u I]^{-1}\right\}.
\end{equation}

We denote this distribution as BYM$(\tau_s, \tau_u)$. This specification allows the model to learn both the strength of the spatial smoothing ($\tau_s$) and the amount of independent observation heterogeneity ($\tau_u$) from the data. The feature effects $\delta$ are given a Gaussian prior, $\delta_c \sim \mathcal{N}(0, \Sigma_\delta)$. The covariance matrix $\Sigma_\delta$ can also be parameterized to capture a shared correlation structure, for instance, using a low--rank representation based on latent factors or a kernel derived from feature similarity metrics. This allows the model to learn feature co-occurrence patterns while maintaining computational efficiency. On the rest of this work we will assume a precision matrix $\Sigma_\delta^{-1}=Q_\delta=(\sigma_\delta^2)^{-1} I_C$.

\begin{theorem}
    The BYM distribution is well--specified, that is, the full conditional distributions generate a valid joint.
\end{theorem}
\begin{proof}
To show this, we need the joint precision matrix for all latent parameters $Q_{\theta}$ to be symmetric and strictly positive definite, as this guarantees a valid, proper Gaussian joint distribution. Let the full parameter vector be $\theta = [\phi, \delta]$. As $\phi$ and $\delta$ are \emph{a priori} independent, the joint precision matrix $Q_{\theta}$ is block--diagonal, $Q_{\theta} = \text{diag}(Q_\phi, Q_\delta)$. This matrix is symmetric because its diagonal blocks are symmetric ---  $Q_\phi = \tau_s L_p + \tau_u I_I$ is symmetric as it is the sum of the symmetric graph Laplacian $L_p$ and the symmetric identity matrix $I_I$, and $Q_\delta = (\sigma_\delta^2)^{-1} I_C$ is diagonal and thus symmetric. 

The matrix $Q_{\theta}$ is also strictly positive definite because its diagonal blocks are strictly positive definite (assuming $\tau_s, \tau_u, \sigma_\delta^2 > 0$) --- $Q_\delta$ is strictly positive definite as its diagonal entries are positive, and $Q_\phi$ is strictly positive definite because it is the sum of a positive semi--definite matrix ($\tau_s L_p$) and a strictly positive definite matrix ($\tau_u I_I$). Since $Q_{\theta}$ is symmetric and strictly positive definite, it has a valid inverse $K_{\theta} = Q_{\theta}^{-1}$, and the joint distribution $p(\vec{\theta}) \sim \mathcal{N}(0, K_{\theta})$ is a proper, well--specified Gaussian. By the Hammersley--Clifford theorem, this joint is consistent with the local conditional specifications.
\end{proof}


The parameters of the BYM model provide an interpretable decomposition of the variance in the latent field.
\paragraph{Covariate--dependent Parameters ($\tau_s, \tau_u$)} These
precision parameters control the observation--level effects.
The first parameter, $\tau_s$,
governs the strength of the spatial smoothing; a large value indicates
that a observation's effect is strongly determined by their
neighbors, which are determined given the baseline covariates, \emph{i.e.}, two items are close if their measured variables are close.
The second parameter, $\tau_u$,  controls the variance of the independent
observation--specific noise. The relative magnitude of these two
parameters is often summarized as a mixing proportion, indicating the
fraction of the total observation--level variance that is attributable
to the structured spatial component. 

\paragraph{Feature Effects ($\delta_c$)} This vector of parameters represents the baseline log-odds for each feature, averaged over all observations. A large positive $\delta_c$ indicates a feature with a high overall prevalence, while a large negative value indicates a rare feature. By examining the posterior distributions of these effects, we can rank features by their prevalence and identify those with significant positive or negative baseline risk. Note that these parameters are not identifiable in the sense that adding any constant $C$ to all of them leaves the likelihood invariant; however we can still analyze the relative differences between them.

The full hierarchical model is specified as:
\begin{align*}
    A_{ic} \mid \Lambda_{ic} &\sim \text{Bernoulli}(\text{logit}^{-1}(\Lambda_{ic})) \\
    \Lambda_{ic} &= \phi_i + \delta_c \\
    \phi &\sim \text{BYM}(\tau_s,\tau_u) \\
    \delta_c &\sim \mathcal{N}(0, \Sigma_\delta)
\end{align*}

\subsection{Adding temporal dynamics}\label{sub:temp}
The separable nature of the hierarchical model allows for
straightforward extensions. For instance, in a time--dynamic setting,
we can introduce a temporal random effect or model the evolution of
condition effects $\delta_c$ or item effects $\phi_i$
over time using a process like AR(1)
or an Ornstein--Uhlenbeck process, allowing the baseline prevalence of
conditions or items to change dynamically.
Keeping in mind the motivating application to polypharmacy data,
we will first focus on the dynamic extension of $\phi_i$ only.
Dynamic extension on $\phi_i$ will allow inference with longitudinal
data, as it is common in biomedical inference.

To extend the model to handle longitudinal data where observations are
recorded over multiple time points, we introduce a time dimension to
the latent field and model the evolution of  item (patient)-specific
effects over time. The observed binary outcome, the feature allocation, for item $i$ regarding
feature $c$ at time $t$ is denoted $A_{itc}$. The logit--probability
is modeled as: 
$$
\text{logit}(p_{itc}) = \Lambda_{itc} = \phi_{it} + \delta_c
$$
Here, the feature effects $\delta_c$ are modeled as time--invariant,
capturing the stable baseline prevalence or importance of each feature
across time. 

The item effects $\phi_{it}$ are now time--varying. At the initial time point ($t=0$), the vector of observation effects $\phi_0$ are given the standard BYM prior described previously, combining a spatially structured ICAR component based on item similarity (e.g., through a k-NN graph $W_p$) and an unstructured component:
$$
\phi_0 \sim \text{BYM}(\tau_s,  \tau_u )
$$
For subsequent time points ($t > 0$), the observation effects evolve according to a discrete-time approximation of an Ornstein--Uhlenbeck (OU) process \citep{PhysRev.36.823}, which is a continuous--time autoregressive process. The effect for item $i$ at time $t$ depends on its effect at the previous time point, $t-1$, and the time elapsed, $\Delta t$:
$$
\phi_{it} \mid \phi_{i, t-1} \sim N(\rho_{ou}^{\Delta t} \phi_{i, t-1}, \sigma_{ou}^2 (1 - \rho_{ou}^{2\Delta t}))
$$
Here, $\rho_{ou}$ is the autocorrelation parameter of the OU process (constrained between 0 and 1), representing the persistence of the effect, and $\sigma_{ou}^2$ is its stationary variance, representing the magnitude of random fluctuations. This temporal structure allows item effects to drift smoothly over time, reflecting potential changes in their underlying state (e.g., a patient's evolving health status). Importantly, the spatial BYM structure (conditional dependence on neighbors) is still enforced at each time step $t$ after the temporal transition. This ensures that the field of item effects remains spatially coherent according to the similarity graph throughout its evolution.

The full hierarchical specification for the dynamic model integrates these components, allowing for simultaneous inference of time-invariant feature effects, potentially correlated feature structure, and smoothly evolving, spatially structured item effects. Priors for the OU parameters $\rho_{ou}$ and $\sigma_{ou}$ are chosen to be weakly informative (e.g., $\rho_{ou} \sim \text{Beta}(a,b)$ and $\sigma_{ou} \sim \text{Half-Normal}$).

\subsection{Prediction}

The structured component of the BYM model is defined on a fixed, finite set of observations, which renders it non-projective. A process is projective if the joint distribution for $N$ data points corresponds to the marginal of the joint distribution for $N+1$ points. In our framework, adding a new item alters the global k--NN graph structure and the resulting precision matrix, requiring a new model definition rather than a simple extension. We deliberately sacrifice projectivity to achieve the significant computational scalability provided by sparse precision matrices.

 However,
we can still efficiently generate predictions for a new, out--of--sample
item, $i^*$, without refitting the model. We employ the
nearest--neighbor interpolation method described by \cite{rue2005},
leveraging the fitted posterior of the latent field. Specifically, we
predict the new item's latent effect, $\phi_{i^*}$, by calculating a
weighted average of the posterior effects of its $k$ nearest
neighbors, $\mathcal{N}(i^*)$. Using inverse-distance weights derived
from the metric $d_{\text{it}}$, the weights are given by: 
\[
w_j = \frac{1/d_{\text{it}}(i^*, j)}{\sum_{l \in \mathcal{N}(i^*)} 1/d_{\text{it}}(i^*, l)} \quad \text{for } j \in \mathcal{N}(i^*).
\]
This interpolation is performed for each sample from the MCMC posterior to generate a full posterior predictive distribution for $\phi_{i^*}$. For each posterior sample $(s)$, we compute:
\[
\phi_{i^*}^{(s)} = \sum_{j \in \mathcal{N}(i^*)} w_j \cdot \phi_{j}^{(s)}.
\]
The final logit--probability for the new item is then $\Lambda_{i^*c}^{(s)} = \phi_{i^*}^{(s)} + \delta_c^{(s)}$, combining the predicted item effect with the sampled feature effect.

This method correctly propagates uncertainty into the predictions and is formally justified by the conditional properties of GMRFs. It provides a fast and practical approach for prediction in large--scale applications where refitting is infeasible.

\section{Simulation Study: False Discovery Rate}

 We assess finite sample properties of inference with the proposed 
model in a simulation study, focusing on
false discovery rate (FDR).
Anticipating the nature of the motivating application, 
the study is designed to evaluate the model's ability to correctly
impute missing data by training it on a partially observed
(``masked'') dataset and assessing the accuracy of its predictions for
the unobserved entries.
Also, for simplicity, we shall refer to items as ``patients'', as it
is the case in the motivating application. 

The simulation proceeds in a series of replications. In each
replication, we first generate a ground-truth dataset. This involves
simulating spatial locations for $I$ patients representing the
correlation induced by baseline covariates, from which we construct a
patient k-nearest neighbors (k-NN) graph and its corresponding graph
Laplacian, $L_p$. We then sample a true patient latent field $\phi$
from the BYM prior using pre-specified hyperparameters for the
structured ($\tau_s$) and unstructured ($\tau_u$) components. A true
vector of condition effects $\delta$ is also sampled. These are
combined to form the true logit-probability surface $\Lambda_{ic} =
\phi_i + \delta_c$, from which a complete binary data matrix
$A_{\text{true}}$ is drawn. 

Next, a training dataset, $A_{\text{train}}$, is created by randomly
masking a fraction of the entries in $a_{\text{true}}$. The BYM model
is then fitted to this incomplete dataset, yielding a posterior
distribution for the latent field $\Lambda$. We use the posterior mean
of $\Lambda$ to make predictions for the 
masked entries. 

A discovery is defined as an instance where the model predicts a probability greater than 0.5 for a masked entry whose true value was 1. 
FDR is then reported as the proportion of predicted ones that were incorrect (i.e., where the true value was 0). This entire process is repeated 100 times to generate a distribution for the FDR, providing an estimate of the model's predictive performance under uncertainty. We estimated an FDR of 37.4\%, which compared with the FDR of 35.9\% of the data--generating processes indicates the model has a good feature discovery performance. The histogram of our estimates is shown in Figure~\ref{fig:fdr}, along with a 2d histogram representing the scatterplot of the true vs estimated posterior mean $\Lambda_{ic}$, which further reinforces the conclusion.

\begin{figure}[t]
    \centering
    \includegraphics[width=0.45\linewidth]{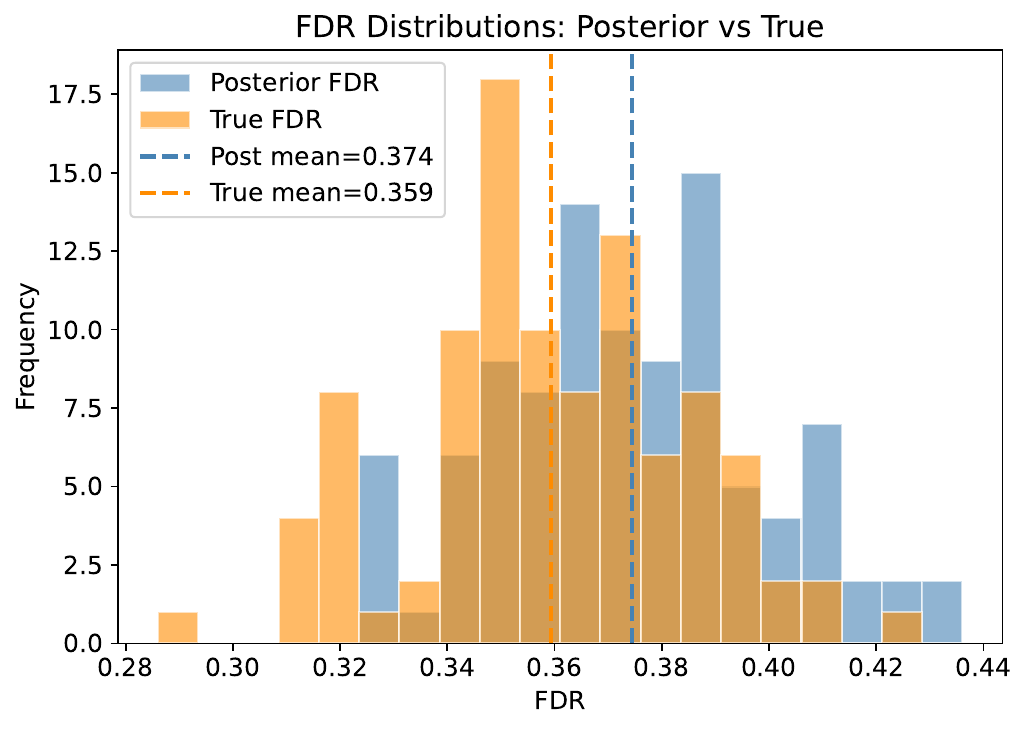}
    \includegraphics[width=0.45\linewidth]{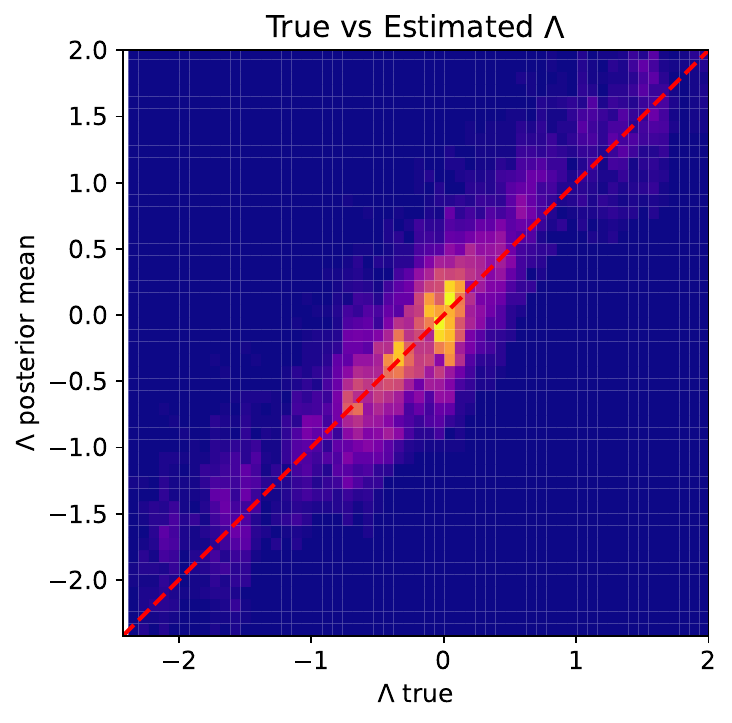}
    \caption{The left plot shows a histogram of the estimated FDR for each simulation using the posterior mean $\Lambda$. The right one is a two--dimensional histogram of $\Lambda$ under the true data generating models and the posterior means.}
    \label{fig:fdr}
\end{figure}

\section{An application to polypharmacy data}\label{sec:app}
\subsection{Polypharmacy data} 
 We apply the proposed model to inference with polypharmacy data.
Polypharmacy, the simultaneous use of multiple
medications, is especially prevalent in patients with chronic
conditions like Human Immunodeficiency Virus (HIV). A primary
challenge in this setting is that a patient's prescription profile may
not fully reflect their underlying health conditions. Medications can
have off-target effects, inadvertently treating the symptoms of a
separate, undiagnosed condition, which makes it 
difficult for clinicians to identify a patient's complete comorbidity
profile. 

We model a patient's underlying condition list as a feature allocation, which we aim to infer from their observed drug prescriptions. The model's task is to learn the latent dependency structure among conditions, allowing us to identify health issues that are otherwise masked by a patient's drug profile.

The indices and variables are presented in Table~\ref{tab:notation}.
\begin{table}[htb]
\centering
\caption{Notation Table}
\label{tab:notation}
\begin{tabular}{l l p{7cm}}
\toprule
Category & Notation & Description \\
\midrule
\multicolumn{3}{l}{Indices} \\
& $i=1,\ldots,I$   & Patient subject index \\
& $t=1, \ldots, T_i$   & Visit index for patient $i$ \\
& $d=1,\ldots,D$   & Drug index \\
& $c=1,\ldots,C$   & Health condition index (known and latent) \\
\midrule
\multicolumn{3}{l}{Latent Variables} \\
& $A_{ic} \in \{0,1\}$ & Latent indicator that patient $i$ has condition $c$. \\
& $\phi_i \in \mathbb{R}$  & Latent effect for patient $i$.\\
& $\delta_c \in \mathbb{R}$  & Latent baseline effect for condition $c$.\\
& $\tau_s, \tau_u, \sigma_\delta$ & Hyperparameters for latent effects. \\
\midrule
\multicolumn{3}{l}{Observed Data and Priors} \\
& $D_i$ & The set of drugs prescribed to patient $i$. \\
&  $B_{dc}$      & Indicator that drug $d$ is prescribed for condition $c$. \\
\midrule
\multicolumn{3}{l}{Model Components} \\
& $d_{\text{pat}}(i,j)$      & Distance measure between patients $i$ and $j$. \\
& $d_{\text{cond}}(c,d)$   & Distance measure between conditions $c$ and $d$.\\
\bottomrule
\end{tabular}
\end{table}

\subsection{Sampling model}

A key challenge in the polypharmacy setting is that the true matrix of patient conditions, $A_{ic}$, is not directly observed. Instead, the available data consists of the drug profiles for each patient, $(D_i)_{i=1}^n$. Our modeling approach therefore treats the condition allocation matrix $A$ as a latent variable to be inferred. The link between this latent structure and the observed data is established through prior knowledge of which drugs are prescribed for which conditions.

First, we construct the set of $C$ possible health conditions that
form the features of our model. This set is grounded by identifying
medications that are uniquely or primarily linked to a single condition
according to established clinical knowledge bases.
 We use the the Anatomical Therapeutic Chemical (ATC) system. 
For instance, a 
drug used exclusively for treating Hepatitis B
essentially reveals this condition. 
The presence of such a drug in a patient's
profile provides strong evidence for the corresponding condition,
helping to ground the inference of the latent space. 

The model connects the observed drug profiles $D_i$ to the latent condition allocations $A_{ic}$ via a likelihood based on a known drug--condition mapping, $B_{dc}$. We assume that if a patient has a condition ($A_{ic}=1$), they are prescribed drugs for it, and the probability of receiving any particular drug $d$ appropriate for that condition is uniform across the set of suitable drugs. To improve mixing we assume a small probability $\epsilon>0$ of being prescribed a drug not indicated for said condition. This forms the observational model.

The latent structure is governed by the hierarchical BYM prior described in Section~5. The logit-probability surface $\Lambda$ determines the probability of the latent condition allocations $A$. The structure of the BYM's spatial component is defined by a patient neighborhood graph, $\mathcal{N}_p$, constructed using a k-nearest neighbors (k-NN) algorithm from the patient distance metric $d_{\text{pat}}(i,j)$. This distance is based on patient covariates and drug profile similarity.

Now, for the condition effects $\delta_c$ we assume a Gaussian prior $N(0, \Sigma_c)$, where $\Sigma_{cc}=\sigma^2_c\sim \text{Half--Normal}(1)$ and $\Sigma_{cc'}=\xi \rho_c\rho_{c'}[\sigma^2_c\sigma^2_{c'}]$, that is, the correlation between condition $c$ and $c'$ is given by the product of a global latent factor $\xi\sim U(-1,1)$ and two individual factors $\rho_c, \rho_{c'}\sim TN(0, 0.01)$, where $TN$ represents a $[-1,1]$--truncated normal. The variance components are given a Half--Normal distribution with variance equal to 1.

The final hierarchical model is specified as:

\begin{align}\label{eq:model}
P(d\in D_i\mid A_{ic}=1)&\overset{\text{ind}}{\propto} \begin{cases} \frac{1-\epsilon}{\sum_{d'} B_{d'c}}\quad &\text{if } B_{dc}=1\\
\epsilon \quad &\text{if } B_{dc}=0 \end{cases} \text{ for all } i,c \\
A_{ic} \mid \phi_i, \delta_c &\sim \text{Bernoulli}(\text{logit}^{-1}(\phi_i + \delta_c)) \nonumber \\
\phi \mid \tau_s, \tau_u &\sim \text{BYM}(\tau_s,\tau_u) \nonumber \\
\delta_c \mid \Sigma_\delta &\overset{\text{iid}}{\sim}  \mathcal{N}(0, \Sigma) \nonumber \\
\Sigma_{cc'} =\xi \rho_c\rho_{c'}[\sigma^2_c\sigma^2_{c'}]\quad &\&\quad
\Sigma_{cc} \sim\text{Half--Normal}(1)\nonumber \\ 
\xi &\sim U(-1,1),\; \rho_c \overset{\text{iid}}{\sim} TN_{[-1,1]}(0,0.01)\nonumber \\
 \tau_s, \tau_u, \sigma_\delta &\sim \text{Half-Cauchy}(2) \nonumber
\end{align}

\subsection{Inference}

A Bayesian treatment of the probability surface in Equation \eqref{eq:model} presents a significant computational challenge. The joint posterior distribution of the latent variables and parameters is high-dimensional and has a complex geometry due to the hierarchical structure of the model. For large datasets, standard MCMC methods like Hamiltonian Monte Carlo (HMC) \citep{Neal2011} become computationally prohibitive, as they require processing the entire dataset for each gradient evaluation.

To overcome this, we employ Wasserstein Consensus Monte Carlo \citep{dunson2015}, a divide-and-conquer strategy designed for scalable Bayesian inference. The core idea is to partition the full dataset $\mathcal{D}$ into $M$ disjoint subsets or shards, $\{\mathcal{D}_1, \ldots, \mathcal{D}_M\}$. We then run an MCMC algorithm (such as HMC) independently and in parallel on each shard, making it more computationally efficient.

This process yields $M$ sets of posterior samples, each corresponding to a sub-posterior distribution $p(\theta | \mathcal{D}_m)$. The key challenge then lies in combining these sub-posteriors to form an accurate approximation of the true full-data posterior, $p(\theta | \mathcal{D})$. This method accomplishes this by computing the Wasserstein barycenter, that is, the measure that minimizes the 2-Wasserstein distance between all the sharded posteriors.

The inference process is as follows:
\begin{enumerate}
    \item Partition the data $\mathcal{D}$ into $M$ shards, $\{\mathcal{D}_1, \ldots, \mathcal{D}_M\}$.
    \item For each shard $m \in \{1, \ldots, M\}$, run an independent MCMC chain to draw samples from the sub-posterior $p(\theta | \mathcal{D}_m)$.
    \item Combine the samples from all sub-posteriors by finding the Wasserstein barycenter, which represents the consensus posterior distribution.
\end{enumerate}

\subsection{Results}

We now present the inference results for the polypharmacy data. We split the data into 14 shards of 500 patients each and ran four parallel chains on each shard using the No-U-Turn sampler (NUTS) implementation from Numpyro \citep{phan2019composable}. Each chain was run for 5000 iterations after a burn-in of 10,000. Further, we use 10 neighbors for each patient in the ICAR component of the prior.

After running the chains we take the barycenter of the model parameters as described in Section 7.2 and recompute the probability surface $\Lambda$. The confusion matrix in Table~\ref{tab:conf}  shows a good classification performance for negative cases (zeros), correctly identifying when a patient does not have a certain condition. The model recovered roughly 56\% of the positive cases (ones). The ground truth was determined by assuming that if a patient takes no prescriptions indicated for a condition then they do not have it, and that assuming they do have it if they take a drug uniquely prescribed for the condition. We also show in Table~\ref{tab:conds} the imputed conditions for three randomly chosen patients.

\begin{table}[h]
    \centering
    \begin{tabular}{l rr}
        \toprule
        & \multicolumn{2}{c}{\textbf{Predicted}} \\
        \cmidrule(l){2-3}
        \textbf{True Class} & \multicolumn{1}{c}{0} & \multicolumn{1}{c}{1} \\
        \midrule
        0 & 0.944 & 0.056 \\
        1 & 0.440 & 0.560 \\
        \bottomrule
    \end{tabular}
        \caption{Confusion matrix based on mean posterior probabilities.}
            \label{tab:conf}
\end{table}

We analyzed the posterior distributions for the condition-specific baseline effects, $\delta_c$. Figure~\ref{fig:betas} shows the posterior mean and the 95\% credible interval for each effect. These parameters represent the overall prevalence of each condition, adjusted for patient similarity. A large positive $\delta_c$ indicates a condition that is common across the population, while a large negative value suggests a rare condition.

\begin{table}[ht]
\centering
\begin{tabular}{|c|p{5cm}|c|}
\hline
\textbf{Patient Index} & \textbf{Original Conditions} & \textbf{Imputed Conditions} \\
\hline
5732 & Anxiety, Depression, Diabetes, HIV, Hyperlipidaemia, Hypertension & Pain (0.53) \\
\hline
6477 & Benign prostatic hyperplasia, HIV, Hyperlipidaemia, Hypertension, Ischaemic heart disease: hypertension, Pain & Depression (0.60) \\
\hline
6426 & Allergies, Anxiety, Chronic airways disease, Congestive heart failure, Diabetes, HIV, Hyperlipidaemia, Hypertension, Ischaemic heart disease: angina, Pain & Depression (0.79) \\
\hline
\end{tabular}
\caption{Selected patients with at least one imputed condition ($p > 0.5$). The original conditions are as recorded in the ground truth; the imputed condition(s) are those predicted by the model at or above the threshold along with their probabilities.}
\label{tab:conds}
\end{table}


\begin{figure}[ht]
    \centering
    \includegraphics[width=\linewidth]{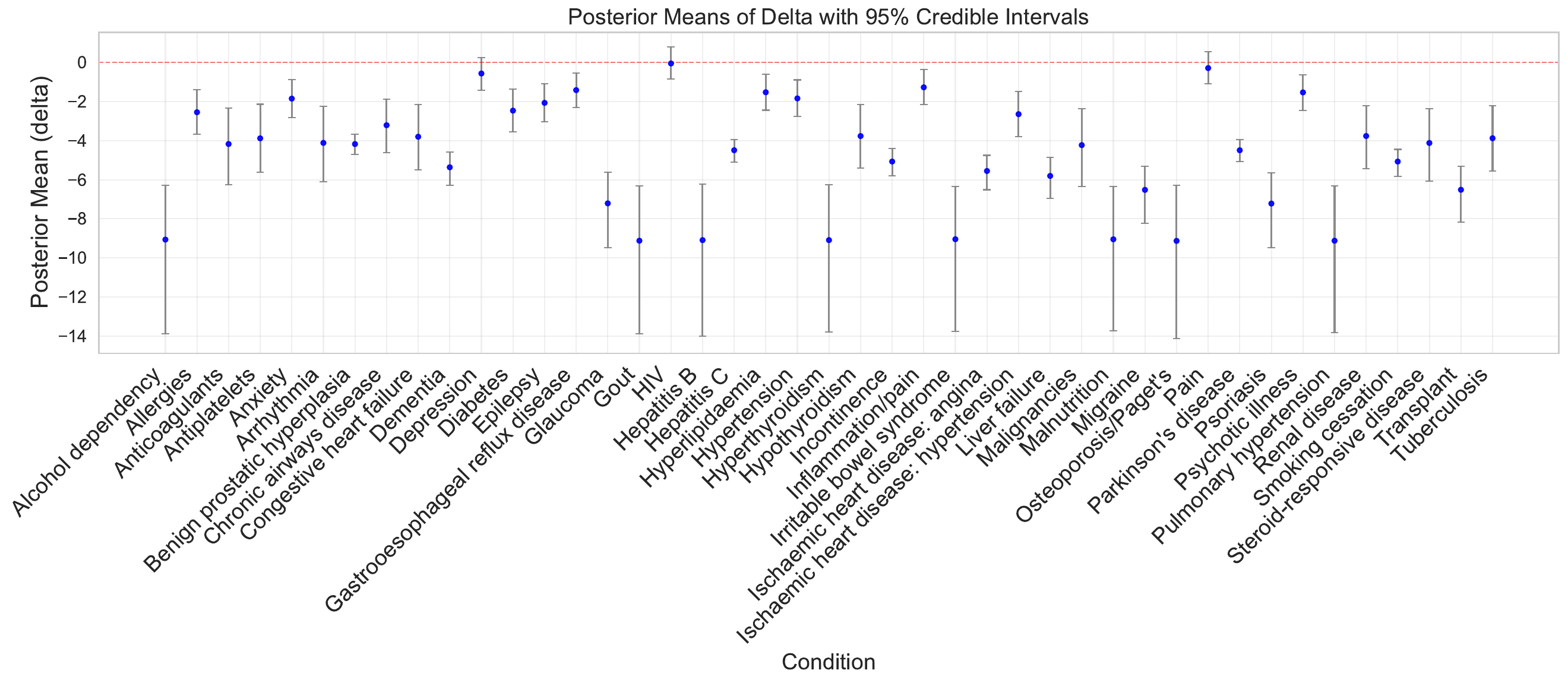}
    \caption{Posterior distributions for the condition baseline effects ($\delta_c$)). The horizontal red line indicates 0.}
    \label{fig:betas}
\end{figure}


The dataset contained records of several follow up visits for each patient with possible missing appointments. Thus we also implemented the time-dynamic model described in Subsection~\ref{sub:temp}. This model extends the static framework by allowing the patient-specific effects, $\phi_{it}$, to evolve according to a discretely--observed Ornstein-Uhlenbeck diffusion. In this process, the patient effect at a given time point, $\phi_{it}$, is modeled as a normal distribution centered on its value from the previous time step, $\phi_{i, t-1}$, scaled by an autocorrelation parameter $\rho_{ou}$. This allows each patient's latent effect to drift smoothly over time, capturing dynamic changes while still being constrained by the spatial BYM structure at each time step.

\begin{figure}[ht]
    \centering
    \includegraphics[width=0.42\linewidth]{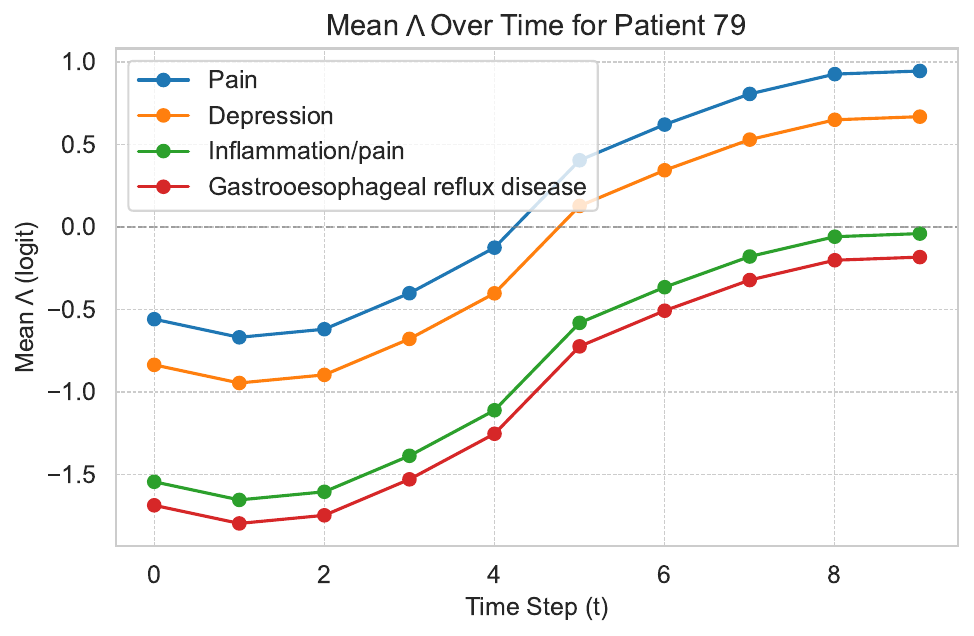}
    \includegraphics[width=0.42\linewidth]{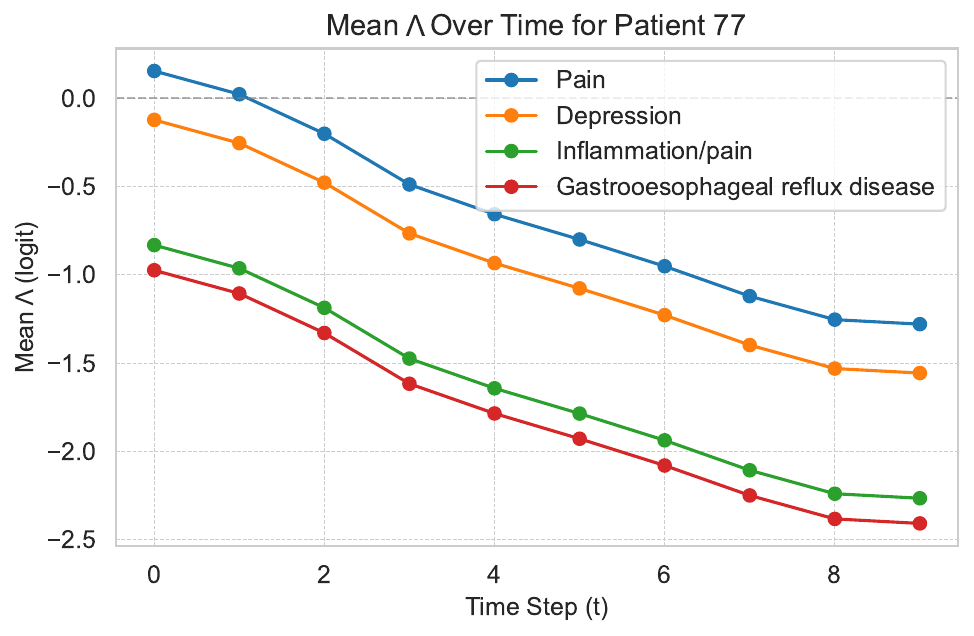}
    \includegraphics[width=0.42\linewidth]{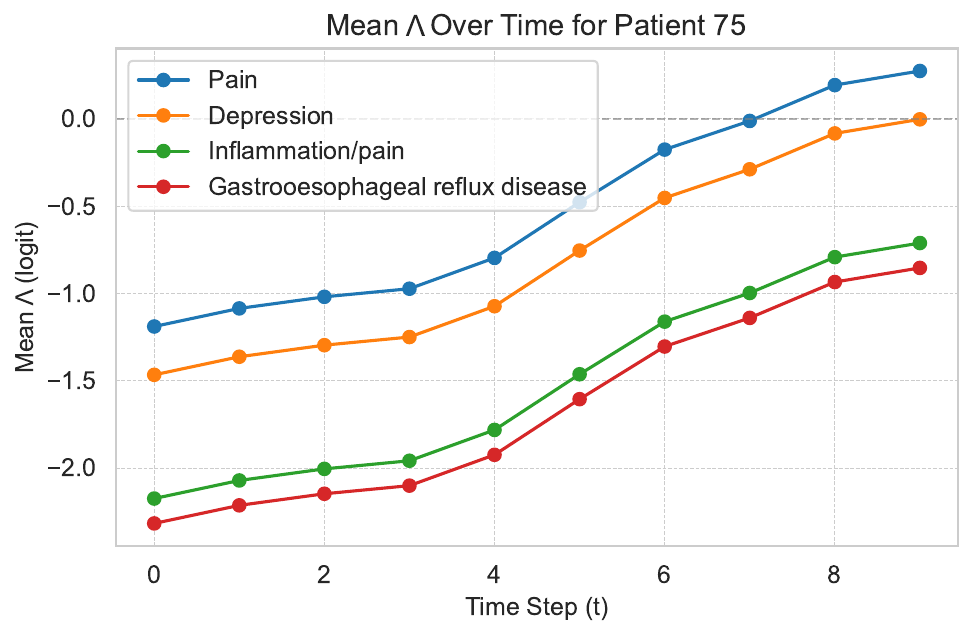}
    \caption{Probability trajectories for three randomly chosen patients for four conditions.}
    \label{fig:traj}
\end{figure}

We specified weakly informative priors for the hyperparameters. The precision parameters for the BYM spatial structure were given $\tau_s \sim \text{HalfCauchy}(2)$ for the structured component and $\tau_u \sim \text{HalfCauchy}(2)$ for the unstructured component. For the OU process we set the autocorrelation parameter $\rho_{ou} \sim \text{Beta}(2, 2)$ and the stationary standard deviation $\sigma_{ou} \sim \text{HalfNormal}(1.0)$. Similarly to the static model, the time-invariant condition effects $\delta$ were modeled jointly with a $\text{MultivariateNormal}(0, \Sigma_\delta)$ prior. This covariance matrix $\Sigma_\delta$ was constructed using the same low--rank structure parameterized by a global scale $\sigma_\delta \sim \text{HalfNormal}(1.0)$, a correlation parameter $\rho \sim TN_{[-1,1]}(0,0.01)$, and a latent factor vector $\xi\sim  \text{Uniform}(-1.0, 1.0)$.

The full hierarchical model for the time--varying patient effects $\phi_t$ is defined as a spatio-temporal GMRF. The prior for the initial state $\phi_0$ is a standard BYM model. The prior for subsequent states $\phi_t$ is a product of the temporal OU transition and the same static BYM potential, which ensures spatial coherence at every time step. 

\begin{align*} 
\phi_0 \mid \tau_s, \tau_u &\sim N(0, Q_0^{-1}) \\
\text{where } Q_0 &= \tau_s L_p + (1 + \tau_u) I_I \\
\text{For } t \geq 1, \quad p(\phi_t \mid \phi_{t-1}) &\propto {\text{OU}}(\phi_t \mid \phi_{t-1}) \otimes p_{\text{spatial}}(\phi_t) \\
\text{where } {\text{OU}}(\phi_t \mid \phi_{t-1})&= N(\phi_t \mid \rho_{ou}^{\Delta t} \phi_{t-1}, \sigma_{dt}^2 I_I) \\
p_{\text{spatial}}(\phi_t) &= \exp\left(-\frac{1}{2} \phi_t' (\tau_s L_p + \tau_u I_I) \phi_t\right) \\
\text{and } \sigma_{dt}^2 &= \sigma_{ou}^2 (1 - \rho_{ou}^{2\Delta t})
\end{align*}

In Figure~\ref{fig:traj}, we show the trajectories for the mean probabilities of three different patients having the top 
 four imputed conditions: pain, depression, reflux, and inflammation. The almost parallel trends indicate these four conditions are strongly correlated with each other.

\section{Discussion}

In this work, we introduced a parametric approach to feature allocation modeling by reconceptualizing it as a spatial problem. By mapping the latent feature paintbox onto a continuous probability surface, we take advantage the rich toolkit of hierarchical spatial modeling to explicitly capture complex dependency structures. Central to this framework is the use of a Besag, York, and Mollié (BYM) decomposition. This approach allows us to rigorously separate the signal-—structured spatial correlations driven by item similarity-—from unstructured heterogeneity. This modularity offers a distinct advantage over traditional non--parametric models, overcoming their often-opaque nature by providing interpretable parameters and straightforward extensibility while maintaining computational tractability through scalable inference. We demonstrated the practical power of this framework through an application to polypharmacy, where the model successfully utilized these latent structures to unmask hidden health conditions within complex, real--world patient profiles.

\section{Appendix I}

The Anatomical Therapeutic Chemical (ATC) code is a classification method for medical drugs, maintained by the World Health Organization Collaborating Centre for Drug Statistics Methodology. It has five levels corresponding to:

\begin{itemize}
    \item 1st level: The system has fourteen main anatomical or pharmacological groups.
    \item 2nd level: Pharmacological or Therapeutic subgroup.
    \item 3rd \& 4th levels: Chemical, Pharmacological or Therapeutic subgroup.
    \item 5th level: Chemical substance.
\end{itemize}

For example, metformin is classified within the ATC system as \texttt{A10BA02}. The first character, \texttt{A}, represents the anatomical main group ``Alimentary tract and metabolism'' (1st level). The next two characters, \texttt{10}, denote the therapeutic subgroup ``Drugs used in diabetes'' (2nd level). Further subdivision to \texttt{A10B} indicates the pharmacological subgroup ``Blood glucose lowering drugs, excluding insulins'' (3rd level). The addition of \texttt{A} (making \texttt{A10BA}) specifies the chemical subgroup ``Biguanides'' (4th level). Finally, the code \texttt{A10BA02} corresponds specifically to the chemical substance metformin (5th level).

\bibliographystyle{ba}
\bibliography{FA}
\end{document}